\documentclass[pra,reprint,aps,twocolumn,twoside,longbibliography,superscriptaddress,nofootinbib]{revtex4-1}
\usepackage{times}
\usepackage[pdftex]{graphicx}
\usepackage[pdftex,colorlinks=true,linkcolor=blue,citecolor=blue,urlcolor=black]{hyperref}
\usepackage{amsmath, amsthm, amssymb}
\usepackage{subfigure}
\usepackage{comment}
\usepackage{url}
\usepackage{tikz}

\newcommand{\C}{\mathbb{C}}

\newcommand{\E}{\mathbb{E}}
\newcommand{\F}{\mathbb{F}}

\newcommand{\ket}[1]{| #1 \rangle}
\newcommand{\bra}[1]{\langle #1|}

\newcommand{\bracket}[3]{\langle #1|#2|#3 \rangle}

\DeclareMathOperator{\poly}{poly}

\DeclareMathOperator{\diag}{diag}

\DeclareMathOperator{\gap}{gap}
\DeclareMathOperator{\ngap}{ngap}

\DeclareMathOperator{\FBPP}{FBPP}

\DeclareMathOperator{\NP}{NP}

\DeclareMathOperator{\GapP}{GapP}

\newcommand{\be}{\begin{equation}}
\newcommand{\ee}{\end{equation}}
\newcommand{\bea}{\begin{eqnarray}}
\newcommand{\eea}{\end{eqnarray}}
\newcommand{\bes}{\begin{equation*}}
\newcommand{\ees}{\end{equation*}}
\newcommand{\beas}{\begin{eqnarray*}}
\newcommand{\eeas}{\end{eqnarray*}}

\makeatletter
\newtheorem*{rep@theorem}{\rep@title}
\newcommand{\newreptheorem}[2]{%
\newenvironment{rep#1}[1]{%
 \def\rep@title{#2 \ref{##1} (restated)}%
 \begin{rep@theorem}}%
 {\end{rep@theorem}}}
\makeatother

\newtheorem{thm}{Theorem}
\newtheorem*{thm*}{Theorem}
\newtheorem{cor}[thm]{Corollary}
\newtheorem{con}[thm]{Conjecture}
\newtheorem{lem}[thm]{Lemma}
\newtheorem*{lem*}{Lemma}
\newtheorem{prop}[thm]{Proposition}

\newtheorem{fact}[thm]{Fact}

\newreptheorem{thm}{Theorem}
\newreptheorem{lem}{Lemma}

\begin{document}

\title{Average-case complexity versus approximate simulation of commuting quantum computations}

\author{Michael J. Bremner}
\affiliation{Centre for Quantum Computation and Intelligent Systems, \\Faculty of Engineering and Information Technology,
University of Technology Sydney, NSW 2007, Australia}
\email{michael.bremner@uts.edu.au}

\author{Ashley Montanaro}
\affiliation{School of Mathematics, University of Bristol, UK}

\author{Dan J. Shepherd}
\affiliation{CESG, Hubble Road, Cheltenham, GL51 0EX, UK}

%\begin{center}
%{\Large \bf Average-case complexity versus approximate simulation of commuting quantum computations}\\
%\bigskip
%{\normalsize Michael J. Bremner$^1$, Ashley Montanaro$^2$ and Dan J. Shepherd$^3$}\\
%\bigskip
%{\small\it $^1$Centre for Quantum Computation and Intelligent Systems,\\
%Faculty of Engineering and Information Technology\\
%University of Technology Sydney, NSW 2007, Australia\\[1mm]
%$^2$Department of Computer Science, University of Bristol, UK\\[1mm]
%$^3$CESG, Hubble Road, Cheltenham, GL51 0EX, U.K.\\}
%\today
%\end{center}

\begin{abstract}
We use the class of commuting quantum computations known as IQP (Instantaneous Quantum Polynomial time) to strengthen the conjecture that quantum computers are hard to simulate classically. We show that, if either of two plausible average-case hardness conjectures holds, then IQP computations are hard to simulate classically up to constant additive error. One conjecture relates to the hardness of estimating the complex-temperature partition function for random instances of the Ising model; the other concerns approximating the number of zeroes of random low-degree polynomials. We observe that both conjectures can be shown to be valid in the setting of worst-case complexity. We arrive at these conjectures by deriving spin-based generalisations of the Boson Sampling problem that avoid the so-called permanent anticoncentration conjecture.
\end{abstract}

\maketitle

% ------------------------------------------------------------------------------

%\section{Introduction}

Quantum computers are conjectured to outperform classical computers for a variety of important tasks ranging from integer factorisation~\cite{shor97} to the simulation of quantum mechanics~\cite{georgescu14}. However, to date there is relatively little rigorous evidence for this conjecture. It is well established that quantum computers can yield an exponential advantage in the query and communication complexity models. But in the more physically meaningful model of time complexity, there are no {\em proven} separations known between quantum and classical computation.

This can be seen as a consequence of the extreme difficulty of proving bounds on the power of classical computing models, such as the famous P vs.\ NP problem. Given this difficulty, the most we can reasonably hope for is to show that quantum computations cannot be simulated efficiently classically, assuming some widely believed complexity-theoretic conjecture. For example, any set of quantum circuits that can implement Shor's algorithm~\cite{shor97} provides a canonical example, with the unlikely consequence of efficient classical simulation of this class of quantum circuits being the existence of an efficient classical factoring algorithm. However, one could hope for the existence of other examples that have wider-reaching complexity-theoretic consequences.

With this in mind, in both \cite{bremner11} and \cite{aaronson13} it was shown that the existence of an efficient classical sampler from a distribution that is close to the output distribution of an arbitrary quantum circuit, to within a small multiplicative error in each output probability, would imply that post-selected classical computation is equivalent to post-selected quantum computation. This consequence is considered very unlikely as it would collapse the infinite tower of complexity classes known as the Polynomial Hierarchy~\cite{papadimitriou94} to its third level. In both works this was proven even for non-universal quantum circuit families: commuting quantum circuits in the case of~\cite{bremner11}, and linear-optical networks in~\cite{aaronson13}.  These non-universal families are of physical interest because they are simpler to implement, and easier to analyse because of the elegant mathematical structures on which they are based.

Unfortunately, this notion of approximate sampling is physically unrealistic: it is more reasonable to allow the quantum computer, and its corresponding classical simulator, to sample from a distribution which is close to the quantum circuit's output distribution in total variation distance (or equivalently the $\ell_1$ distance). The results of~\cite{bremner11,aaronson13} have little meaning in this ``additive error'' setting.

%However, in physically realistic scenarios where the quantum computer and its corresponding classical simulator are allowed to be accurate up to a small additive error, these results have little meaning.

One important recent step to addressing this was proposed by Aaronson and Arkhipov \cite{aaronson13}, who gave a sophisticated argument based on counting complexity that approximately sampling from the output probability distribution of a randomly chosen network of noninteracting photons (a problem known as Boson Sampling) should be classically hard, even up to a reasonable \emph{additive} error bound. This major breakthrough rests on two tantalising but as yet unproven conjectures: the so-called \emph{permanent anticoncentration conjecture} and the \emph{permanent-of-Gaussians conjecture}.

%----------------------------------------------------------------------------------

%\subsection{Our results}
%In this paper we propose a generalisation of the Boson Sampling argument of \cite{aaronson13} that is native to spin systems, specifically to the class of commuting quantum circuits known as IQP (Instantaneous Quantum Polynomial time), which was introduced in \cite{shepherd09} and \cite{bremner11}. In our opinion this leads to a mathematically simpler setting, while still apparently retaining the essential complexity-theoretic ingredients. This simplicity allows us to prove the IQP analogues of the \emph{permanent anticoncentration conjecture}. The only remaining conjecture to prove is an IQP analogue of the \emph{permanent-of-Gaussians} conjecture, of which we find two natural examples. 

In this paper we propose a generalisation of the Boson Sampling argument of \cite{aaronson13} that is native to spin systems, specifically to randomly chosen commuting quantum circuits from the class known as IQP (Instantaneous Quantum Polynomial time), which was introduced in \cite{shepherd09} and \cite{bremner11}. In our opinion this leads to a mathematically simpler setting, while still retaining the essential complexity-theoretic ingredients. This simplicity allows us to prove the IQP analogues of the \emph{permanent anticoncentration conjecture}. The only remaining conjecture to prove is an IQP analogue of the \emph{permanent-of-Gaussians} conjecture, of which we find two natural examples.

%Informally (see~\cite{bremner11} for the formal definition) an $n$-qubit IQP circuit $\mathcal{C}$ is a quantum circuit which takes as input the state $|0\rangle^{\otimes n}$, whose gates are diagonal in the Pauli-X basis, and whose $n$-qubit output is measured in the computational basis. We say that a classical sampler of an IQP circuit is accurate up to error $\epsilon$ in $\ell_1$ norm if its output probability distribution has total variation distance at most $\epsilon/2$ from that of $\mathcal{C}$. We can now state the main result of this paper:

Informally (see~\cite{bremner11,shepherd09} for the formal definition) an $n$-qubit IQP circuit $\mathcal{C}$ is a quantum circuit which takes as input the state $|0\rangle^{\otimes n}$, whose gates are diagonal in the Pauli-X basis, and whose $n$-qubit output is measured in the computational basis. It is often convenient to write $\mathcal{C}=H^{\otimes n} \widetilde{\mathcal{C}} H^{\otimes n}$, where $\widetilde{\mathcal{C}}$ is diagonal in the Pauli-Z basis and $H$ is the usual Hadamard gate.  We can now state the main result of this paper:

\begin{thm}
\label{thm:main}
Assume either Conjecture \ref{con:Ising} or \ref{con:deg3} below is true. If it is possible to classically sample from the output probability distribution of any IQP circuit $\mathcal{C}$ in polynomial time, up to an error of 1/192 in $\ell_1$ norm, then there is a $\text{BPP}^{\text{NP}}$ algorithm to solve any problem in $\text{P}^{\#\text{P}}$. Hence the Polynomial Hierarchy would collapse to its third level.\end{thm}

Loosely speaking, the complexity class $\text{P}^{\#\text{P}}$ appearing in this theorem is the class of problems that can be solved in polynomial time given the ability to count the number of solutions of arbitrary NP problems~\cite{papadimitriou94}; $\text{BPP}^{\text{NP}}$ is the class of problems that can be solved by randomised classical polynomial-time computation equipped with an oracle that can solve any problem in NP. 

%Theorem \ref{thm:main} is based around two natural conjectures, one of which is native to computer science, and other common in condensed-matter physics. Each conjecture straightforwardly relates to a family of IQP circuits. Essentially the result states that if we assume either of our conjectures, then there is \emph{no way of efficiently simulating these families of quantum circuits} without a major re-evaluation of the existing status-quo of complexity theory. 

Theorem \ref{thm:main} is based around two natural conjectures, one of which is native to computer science, and other common in condensed-matter physics. Both conjectures are concerned with the complexity of computing multiplicative approximations to output probabilities, $|\langle 0 |^{\otimes n}\mathcal{C}|0\rangle^{\otimes n}|^2$, of circuits that are randomly chosen from circuit families within IQP.  We say that $A$ is a \emph{multiplicative approximation} of $X$ to within $\eta$ if $|A-X|\leq \eta X$. Essentially, Theorem \ref{thm:main} states that if we assume either of our conjectures, then there is \emph{no way of classically efficiently sampling the outputs of these families of quantum circuits} without a major re-evaluation of the existing status-quo of complexity theory. 

%The first conjecture is based on the complexity of one of the most commonly studied models of statistical physics, the Ising model. Consider the partition function
%%
%\begin{equation}
%\label{eq:partition}
%Z(\omega) = \sum_{z \in \{\pm1\}^n} \omega^{\sum_{i<j} w_{ij} z_i z_j + \sum_{k=1}^n v_k z_k }, \end{equation}
%%
%where the exponentiated sum is over the complete graph on $n$ vertices, $w_{ij}$ and $v_k$ are real edge and vertex weights, and $\omega \in \C$. Then, for any $\omega = e^{i \theta}$, $Z(\omega)$ arises straightforwardly as an amplitude of some IQP circuit $\mathcal{C}$: $\langle 0|^{\otimes n}\mathcal{C}|0\rangle^{\otimes n} = Z(\omega)/2^n$ (see Appendix \ref{appendix:Ising} and \cite{shepherd10,iblisdir12,ni12,fujii13,goldberg14}). 

The first conjecture is based on the complexity of one of the most commonly studied models of statistical physics, the Ising model. Consider the partition function
\begin{equation}
\label{eq:partition}
Z(\omega) = \sum_{z \in \{\pm1\}^n} \omega^{\sum_{i<j} w_{ij} z_i z_j + \sum_{k=1}^n v_k z_k }, \end{equation}
where the exponentiated sum is over the complete graph on $n$ vertices, $w_{ij}$ and $v_k$ are real edge and vertex weights, and $\omega \in \C$. Then, for any $\omega = e^{i \theta}$, $Z(\omega)$ arises straightforwardly as an amplitude of some IQP circuit $\mathcal{C}_I(\omega)$: $\langle 0|^{\otimes n}\mathcal{C}_I(\omega)|0\rangle^{\otimes n} = Z(\omega)/2^n$ (see Appendix \ref{appendix:Ising} and \cite{shepherd10,iblisdir12,ni12,fujii13,goldberg14}). For our purposes it is sufficient to restrict to the case where $\omega = e^{i\pi/8}$ and the weights are picked by choosing uniformly at random from the set $\{0,\dots,7\}$ on the vertices and edges of the complete graph on $n$ vertices.  Our results would still apply for many other choices for $\omega$ and the weights (for example, the edge weights can be picked uniformly from $\{0,\dots,3\}$).

The diagonal component of the corresponding circuits $\mathcal{C}_I(e^{i\pi/8})$ can be constructed from $\sqrt{\text{CZ}}$ (square-root of Controlled-Z, i.e.\ $\diag(1,1,1,i)$), and $T=\left( \begin{smallmatrix} 1&0\\0 & e^{i\pi/4} \end{smallmatrix} \right)$ gates, or alternatively by applying the Ising interaction directly. The number of applications of each gate is given by a simple function of the edge and vertex weights of the associated graph in such a way that random edge weights correspond to a random circuit (see the end of Appendix~\ref{sec:tech}). See Figure \ref{fig:4qubit} for an example. Let $Z_R$ denote partition functions associated with this random choice of weights.

%Such partition functions arise naturally from circuits with Ising-interaction gates such as $\omega^{X\otimes X}, \omega^{X}$ where the weights correspond to the number of applications of these gates. IQP circuits can always be written as an $n$-fold tensor product of a Hadamard gate $H^{\otimes n}$, a component whose gates are diagonal in the Pauli-Z basis, and a further application of $H^{\otimes n}$. In this case, we can consider the diagonal component in terms of commonly used  gates,  Z, $\sqrt{\text{CZ}}$, and $T=\left( \begin{smallmatrix} 1&0\\0 & e^{i\pi/4} \end{smallmatrix} \right)$ and we see that the circuits associated with $Z_R(e^{i \pi /8})$ are those that are chosen at random from this gate set. Then our conjecture is:
%\begin{defn}
%\label{defn:Ising}
%$Z_R(e^{i \pi /8})$ is the partition function of a random instance of the Ising model, picked by choosing uniformly random weights from the set $\{0,\dots,15\}$ on the vertices and edges of the complete graph on $n$ vertices.
%\end{defn}

%Our conjecture is:

\begin{con}
\label{con:Ising}
It is \#P-hard to approximate $|Z_R|^2$ up to multiplicative error $1/4 + o(1)$ for a $1/24$ fraction of instances over the choice of vertex and edge weights.
\end{con}

It is known that the family of partition functions $Z(\omega)$ parametrised as above is \#P-hard to compute in the worst case up to the above multiplicative error bound~\cite{fujii13,goldberg14}. Conjecture \ref{con:Ising} thus states that this worst-case hardness result can be improved to an average-case hardness result. 

\begin{figure}[!t]
\includegraphics[width=\columnwidth]{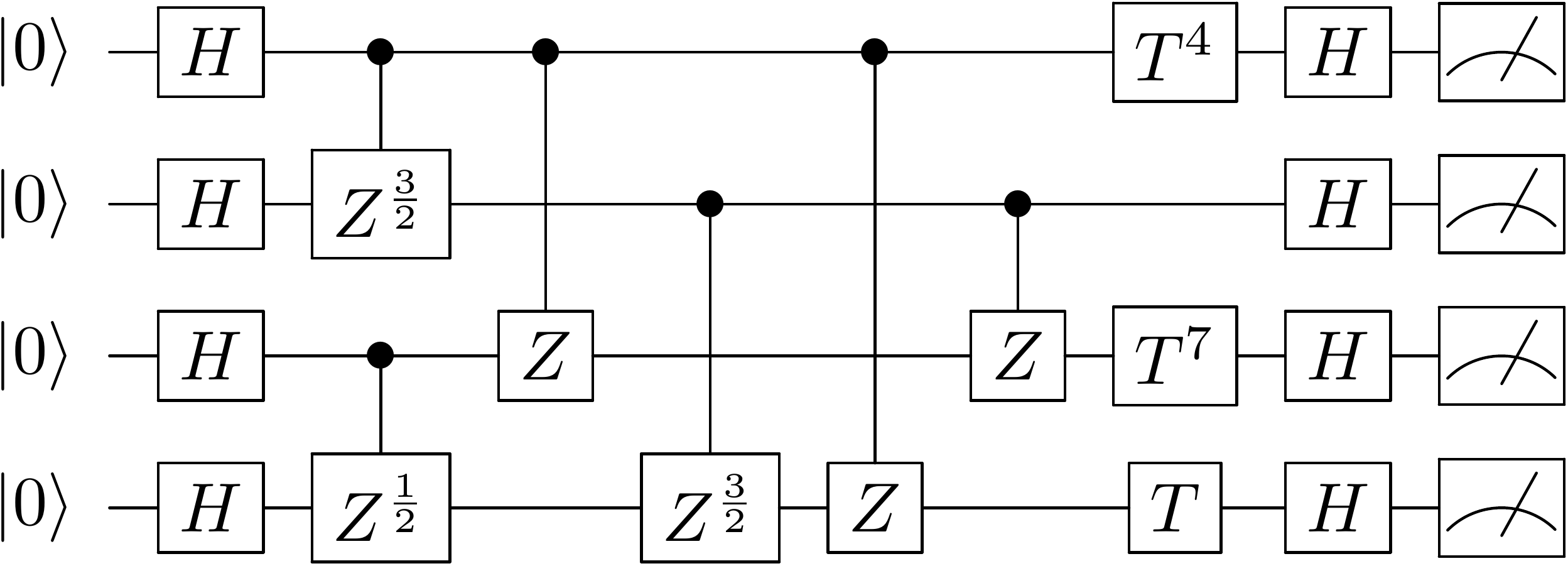}
\caption{An example of a randomly chosen circuit $\mathcal{C}_I$ corresponding to a 4-qubit Ising model instance such that $\langle 0|^{\otimes n} \mathcal{C}_I|0\rangle^{\otimes n} = Z_R/2^n$ (up to trivial phase factors). Assuming Conjecture \ref{con:Ising} is true, if there exists a classically efficient algorithm for sampling from the output of any such ($n$-qubit) circuit to within a constant additive error, then the Polynomial Hierarchy collapses.} \label{fig:4qubit}
\end{figure}
% -----------------------------------------------------------------------

Our second conjecture is based on the hardness of computing the \emph{gap} of degree-3 polynomials over $\mathbb{F}_2$, $f:\{0,1\}^n \rightarrow \{0,1\}$, which are expressible (up to an additive constant) as
\[ f(x) = \sum_{i,j,k} \alpha_{ijk} x_i x_j x_k + \sum_{i,j} \beta_{ij} x_i x_j + \sum_i \gamma_i x_i\;\;\text{(mod 2)}, \]
where $\alpha_{ijk}, \beta_{ij}, \gamma_i \in \{0,1\}$. The gap is defined by $\gap(f) := |\{x: f(x)=0\}| - |\{x: f(x)=1\}|$. It can be shown that, for any degree-3 polynomial $f$, $\langle 0|^{\otimes n}\mathcal{C}_{f} |0\rangle^{\otimes n} = \gap(f)/ 2^n$ for IQP circuits $\mathcal{C}_f$ whose diagonal component is constructed from $Z$, $CZ$, and $CCZ$ gates for the degree 1-3 terms respectively (see Appendix \ref{sec:deg3}). We write $\ngap(f)= \gap(f)/2^n$. Then we have the following conjecture:

\begin{con}
\label{con:deg3}
Let $f:\{0,1\}^n \rightarrow \{0,1\}$ be a uniformly random degree-3 polynomial over $\F_2$. Then it is \#P-hard to approximate $\ngap(f)^2$ up to a multiplicative error of $1/4 + o(1)$ for a $1/24$ fraction of polynomials $f$.
\end{con}

It has been known for some time that $\ngap(f)$ is \#P-hard to compute exactly in the worst case~\cite{ehrenfeucht90}. We show in Appendix~\ref{sec:multiplicative}, using IQP techniques, that this worst-case hardness still holds for approximating $\ngap(f)^2$ up to multiplicative error less than $1/2$. Just as with Conjecture \ref{con:Ising}, what remains is to lift this worst-case hardness result to average-case hardness.

We remark that an additional piece of evidence provided in~\cite{aaronson13} for the average-case hardness of multiplicative approximations to Boson Sampling (the permanent-of-Gaussians conjecture) was a direct proof that exact simulation of Boson Sampling probability distributions is hard on average. This was based on average-case hardness results for computation of the permanent, for which we do not know IQP analogues. However, currently known techniques do not seem sufficient to extend these exact average-case hardness results for Boson Sampling to approximate hardness results~\cite[Section 9.2]{aaronson13}. 

As with the case of Boson Sampling, the worst-case hardness of multiplicative approximations to both  $Z(\omega)$ and $\ngap(f)$ (\cite{goldberg14} and Appendix \ref{sec:multiplicative}) implies via standard results on random-self-reducibility \cite{feigenbaum93} that there exists \emph{some} distribution over the choices of these functions that is \#P-hard on average -- but not necessarily the distributions that we require for Conjectures \ref{con:Ising} and \ref{con:deg3}.

Interestingly, recent independent work of Fefferman and Umans~\cite{fefferman15} has explored an alternative way to generalise the ideas of Aaronson and Arkhipov~\cite{aaronson13}. This work uses Quantum Fourier Sampling to construct states whose corresponding probability distributions are hard to sample from classically, under similar conjectures to~\cite{aaronson13}. An appealing aspect of the construction of~\cite{fefferman15} is that it shows that there are specific, and rather simple, quantum states which are hard to simulate classically, assuming an anticoncentration conjecture holds. However, constructing these states appears to require the full power of quantum computation, unlike the results described here and in~\cite{aaronson13}.

% ------------------------------------------------------------------------------

%\subsection{Preliminaries}

\textbf{Intuition --} There are a number of technical ingredients of Theorem \ref{thm:main} which will be discussed in full below. The basic idea is that, for the class of problems underlying Conjectures \ref{con:Ising} and \ref{con:deg3}, any classical IQP sampler that is accurate up to a good additive error bound in the worst case, is forced to also be accurate to within a reasonable multiplicative error on average. This observation is combined with a classic result of complexity theory, the so-called Stockmeyer counting algorithm (\cite{stockmeyer85} and Appendix \ref{sec:addmult}), which can be used to estimate individual output probabilities of a classical sampler up to small multiplicative error.

Some ingredients essential to the result are anticoncentration results for $\ngap(f)$ (for Conjecture \ref{con:deg3}) and the partition function of the random Ising model (for Conjecture \ref{con:Ising}). That such anticoncentration results can be proven is a consequence of the elegant mathematical structures upon which IQP circuits are based.

%We now describe our results. We assume basic knowledge of standard complexity classes~\cite{nielsen00,papadimitriou94} such as P, \#P and BPP. A crucial complexity-theoretic tool we will use is a result of Stockmeyer on approximate counting:

%\begin{thm}[Stockmeyer~\cite{stockmeyer85}, see~\cite{aaronson13} for statement here]
%\label{thm:stockmeyer}
%There exists an $\FBPP^{\NP}$ machine which, for any boolean function $f:\{0,1\}^n \rightarrow \{0,1\}$, can approximate
%%
%\[ p = \Pr_x [ f(x) = 1 ] = \frac{1}{2^n} \sum_{x \in \{0,1\}^n} f(x) \]
%%
%to within a multiplicative factor of $(1+\epsilon)$ for all $\epsilon = \Omega(1/\poly(n))$, given oracle access to $f$.
%\end{thm}

Putting these observations together, we find that there is an $\text{FBPP}^\text{NP}$ algorithm for computing a multiplicative approximation to a large fraction of $|Z_R|^2$ and $\ngap(f)^2$. Assuming the Conjectures \ref{con:Ising} and \ref{con:deg3}, and that the Polynomial Hierarchy does not collapse, this implies that randomly chosen circuits from $\mathcal{C}_I$ and $\mathcal{C}_f$ cannot be classically simulated. 

We remark that the gate sets $\mathcal{C}_I$ and $\mathcal{C}_f$ would be universal if we could also perform Hadamard (H) gates at any point in the circuit -- which we cannot do in IQP because this gate does not commute with the X gate. However, if we allow the unphysical resource of postselection, these Hadamard gates can effectively be implemented~\cite{bremner11}, allowing IQP circuit amplitudes $\bracket{y}{\mathcal{C}}{x}$ to express any quantum circuit amplitude (up to a known constant). See Appendix~\ref{sec:hadgadget} for a description of this construction. A consequence is that classical computation of such amplitudes is \#P-hard, even up to constant multiplicative error~\cite{bremner11, goldberg14} or exponentially small additive error~\cite{ni12}.

% ------------------------------------------------------------------------------

%\section{Approximation of general IQP circuits}

\textbf{Approximation of general IQP circuits --} We first prove a key technical ingredient, which relates approximate sampling from the output distributions of IQP circuits to approximating individual output probabilities. This is essentially the same argument as used in~\cite{aaronson13} for the permanent, although we believe it becomes substantially simpler in the setting of IQP. The intuition behind this result is that adding random X gates to an IQP circuit randomly permutes the output probabilities. This allows the user of a sampler which is accurate for all circuits to obfuscate from the sampler which one of the output probabilities the user is interested in.

\begin{lem}
\label{lem:approx}
Let $\mathcal{C}$ be an arbitrary IQP circuit on $n$ qubits. Let $\mathcal{C}_x$, for $x \in \{0,1\}^n$, be the circuit produced by appending an X gate to $\mathcal{C}$ for each $i$ such that $x_i = 1$. Assume there exists a classical polynomial-time algorithm $\mathcal{A}$ which, for any IQP circuit $\mathcal{C}'$, can sample from a probability distribution which approximates the output probability distribution of $\mathcal{C}'$ up to additive error $\epsilon$ in $\ell_1$ norm. Then, for any $\delta$ such that $0 < \delta < 1$, there is a $\FBPP^{\NP}$ algorithm which, given access to $\mathcal{A}$, approximates $|\bracket{0}{\mathcal{C}_x}{0}|^2$ up to additive error
\[ O((1+o(1))\epsilon/(2^{n}\delta) + |\bracket{0}{\mathcal{C}_x}{0}|^2/\poly(n)) \]
with probability at least $1-\delta$ (over the choice of $x$).
\end{lem}

\begin{proof}
For each $x,y \in \{0,1\}^n$, set
\[ p_{xy} := \Pr[\mathcal{C}_x \text{ outputs } y], q_{xy} := \Pr[\mathcal{A} \text{ outputs } y \text{ on input } \mathcal{C}_x], \]
where in both cases the probability is taken over the algorithm's internal randomness. For any $y$ of our choice, we can apply Stockmeyer's Counting Theorem (see Appendix \ref{sec:addmult}) to estimate $q_{0y}$ (where we write 0 for $0^n$). This gives us an $\FBPP^{\NP}$ algorithm which produces an estimate $\widetilde{q}_y$ such that
\[ |\widetilde{q}_y - q_{0y}| \le q_{0y}/\poly(n). \]
Then
\beas
|\widetilde{q}_y - p_{0y}| &\le& |\widetilde{q}_y - q_{0y}| + |q_{0y} - p_{0y}|\\
&\le& q_{0y}/\poly(n) + |q_{0y} - p_{0y}|\\
&\le& (p_{0y} + |q_{0y}-p_{0y}|)/\poly(n) + |q_{0y} - p_{0y}|\\
&=& p_{0y}/\poly(n) + |q_{0y}-p_{0y}|(1+1/\poly(n)).
\eeas
As $\mathcal{A}$ approximates the output probability distribution of $\mathcal{C}_0$ up to $\ell_1$ error $\epsilon$, it follows from Markov's inequality that
\[ \Pr_y[ | p_{0y} - q_{0y}| \ge \epsilon/(2^n\delta) ] \le \delta \]
for any $0 < \delta < 1$, where $y$ is picked uniformly at random. Thus
\[ |\widetilde{q}_y - p_{0y}| \le p_{0y}/\poly(n) + \frac{\epsilon(1+1/\poly(n))}{2^{n}\delta} \]
with probability at least $1-\delta$, over the choice of $y$. But
%
%\[ p_{0y} = \Pr[\mathcal{C}_0 \text{ outputs } y] = \Pr[\mathcal{C}_y \text{ outputs } 0] = |\bracket{0}{\mathcal{C}_y}{0}|^2. \]
\[ p_{0y} = \Pr[\mathcal{C}_0 \text{ outputs } y] = |\bracket{y}{\mathcal{C}_0}{0}|^2 = |\bracket{0}{\mathcal{C}_y}{0}|^2. \]
This completes the proof.
\end{proof}

If $|\bracket{0}{\mathcal{C}_x}{0}|^2 = \Omega(2^{-n})$, the algorithm of Lemma \ref{lem:approx} gives a good approximation -- i.e.\ a multiplicative approximation to within roughly $O(\epsilon)$. We state this formally, and calculate the precise constants involved, in Appendix~\ref{sec:addmult}.

We next show that this condition is indeed satisfied for two interesting families of IQP circuits.

% ------------------------------------------------------------------------------

%\section{Anticoncentration bounds}
%\label{sec:anti}

\textbf{Anticoncentration bounds --} Fix a family $\mathcal{F}$ of IQP circuits. We would like to show that $|\bracket{0}{\mathcal{C}}{0}|^2$ is likely to be high for a circuit $\mathcal{C}$ formed by picking a random circuit $\mathcal{D}$ from $\mathcal{F}$, then appending X gates on a uniformly random subset $S$ of the qubits. We will use the following fact:

\begin{fact}[Paley-Zygmund inequality]
\label{fact:pz}
If $R$ is a non-negative random variable with finite variance, then for any $0<\alpha < 1$,
$\Pr[ R \ge \alpha\,\E[R] ] \ge (1-\alpha)^2 \E[R]^2/\E[R^2].$
\end{fact}
We will apply Fact \ref{fact:pz} to the random variable $R = |\bracket{0}{\mathcal{C}}{0}|^2$, first observing that 
\bea
\E_{\mathcal{C}}[|\bracket{0}{\mathcal{C}}{0}|^2] &=& \E_{\mathcal{D},x}[|\bracket{x}{\mathcal{D}}{0}|^2] \nonumber \\
&= &\frac{1}{2^n} \E_{\mathcal{D}}\sum_{x \in \{0,1\}^n} |\bracket{x}{\mathcal{D}}{0}|^2 = \frac{1}{2^n}, \nonumber \eea
where in the second expectation $x$ is picked uniformly at random from $\{0,1\}^n$. This deals with the numerator; to handle the denominator, we need to upper-bound $\E[|\bracket{0}{\mathcal{C}}{0}|^4]$.

% ------------------------------------------------------------------------------

%\subsection{Polynomials over $\F_2$}

The first family of circuits we consider, $\mathcal{C}_f$, corresponds to polynomials over $\F_2$. We prove in Appendix \ref{sec:tech} that for uniformly random degree-3 polynomials $f$, $\E_f[\ngap(f)^4] \le 3 \cdot 2^{-2n}$. Based on this, and the tight connection between IQP circuits over the gate set $\{\text{Z}, \text{CZ}, \text{CCZ}\}$ and degree-3 polynomials, we have the following result:

\begin{thm}
\label{thm:deg3avg}
Assume there exists a classical polynomial-time algorithm $\mathcal{A}$ which, for any IQP circuit $\mathcal{C}$, can sample from a probability distribution which approximates the output probability distribution of $\mathcal{C}$ up to additive error $1/192$ in $\ell_1$ norm. Then there is an $\FBPP^{\NP}$ algorithm which, given access to $\mathcal{A}$, approximates $\ngap(f)^2$ up to multiplicative error $1/4 + o(1)$ on at least a $1/24$ fraction of degree-3 polynomials $f:\{0,1\}^n \rightarrow \{0,1\}$.
\end{thm}

\begin{proof}
Combining Fact \ref{fact:pz} and the bound on $\E_f[\ngap(f)^4]$, we have $\Pr_f[\ngap(f)^2 \ge \alpha/ 2^n] \ge (1-\alpha)^2/3 $ for any $0 < \alpha < 1$. Fixing $\alpha = 1/2$, we get $\Pr_f[ \ngap(f)^2 \ge 2^{-n-1} ] \ge 1/12$. The claim then follows from the discussion above (where the precise parameter values stated in the theorem follow from Corollary \ref{cor:addapprox} in Appendix~\ref{sec:addmult}).
% The claim then follows by inserting these parameters into Corollary \ref{cor:addapprox}.
\end{proof}

% ------------------------------------------------------------------------------

%\subsection{The Ising model}

We next consider the Ising model, where we are interested in evaluating the partition function $Z_R$ for a randomly weighted graph (see (\ref{eq:partition})). Recall each edge of the complete graph has a weight $w_{ij}$, and each vertex has a weight $v_k$, each picked uniformly at random from the set $\{0,\dots,7\}$.

We show in Appendix \ref{sec:tech} that $\bracket{0}{\mathcal{C}_I}{0} = Z_R/2^n$ for an IQP circuit $\mathcal{C}_I$ whose diagonal component is picked from the set $\{\diag(1,1,1,i), \diag(1,e^{i\pi / 4})\}$ (up to trivial phase factors); and further that we can consider a random circuit of this form as being chosen by picking a random circuit using this gate set, then following it by a random choice of X gates. In addition, $\E_{w,v} \left[|Z_R|^4 \right] \le 3\cdot 2^{2n}$. Via Fact \ref{fact:pz} this implies the following result, whose proof is essentially the same as that of Theorem \ref{thm:deg3avg}:

\begin{thm}
\label{thm:isingavg}
Assume there exists a classical polynomial-time algorithm $\mathcal{A}$ which, for any IQP circuit $\mathcal{C}$, can sample from a probability distribution which approximates the output probability distribution of $\mathcal{C}$ up to additive error $1/192$ in $\ell_1$ norm. Then there is a $\FBPP^{\NP}$ algorithm which, given access to $\mathcal{A}$, approximates $|Z_R|^2$ up to multiplicative error $1/4 + o(1)$ with probability at least $1/24$ (over the choice of weights).
\end{thm}

%\begin{proof}
%By a similar argument to the proof of Theorem \ref{thm:deg3avg},
%
%\[ \Pr_{w,v}\left[ |Z(e^{i \pi /8})|^2 \ge \frac{\alpha}{2^n} \right] \ge \frac{(1-\alpha)^2}{3} \]
%
%for any $0 < \alpha < 1$. Fixing $\alpha = 1/2$, we get $\Pr_f[ |Z(e^{i \pi /8})|^2 \ge 2^{-n-1} ] \ge 1/12$. The claim then follows by inserting these parameters into Corollary \ref{cor:addapprox}.
%\end{proof}

Combining Theorems \ref{thm:deg3avg} and \ref{thm:isingavg} gives our main result, Theorem \ref{thm:main}. 

% ------------------------------------------------------------------------------

%\section{Outlook}

%\textbf{Implementation and outlook --} We have argued that the following ``simple'' quantum algorithm should be classically intractible: \textbf{(1)} preparing the computational basis state $|0\rangle^{\otimes n}$, \textbf{(2)} evolving by a circuit randomly drawn from either $\mathcal{C}_I$ or $\mathcal{C}_f$, \textbf{(3)} measuring all $n$ qubits in the computational basis, and \textbf{(4)} repeating $(1)-(3)$ polynomially many times.  In our opinion, the most compelling question is how best to physically implement the classes of IQP circuits discussed in this paper, as they are arguably among the simplest classes of quantum computations that are unlikely to be classically simulable.

\textbf{Implementation and outlook --} We have argued that the following 'simple' quantum algorithm should be classically intractable: \textbf{(1)} preparing the computational basis state $|0\rangle^{\otimes n}$, \textbf{(2)} evolving by a circuit, or equivalent Hamiltonian, randomly drawn from either $\mathcal{C}_I$ (e.g. see Figure \ref{fig:4qubit}) or $\mathcal{C}_f$, \textbf{(3)} measuring all $n$ qubits in the computational basis, and \textbf{(4)} repeating \textbf{(1)}-\textbf{(3)} polynomially many times. By `classically intractable' we mean that if this process can be demonstrated in the laboratory, if the total effect of all errors can be demonstrated to remain consistently below 1/192 (in $\ell_1$ distance) even as the complexity parameter increases, and if the resources (time) the experiment takes can be argued to grow only polynomially with the complexity parameter, then the process is actively evidencing violation of the extended Church-Turing thesis.  Note also that it doesn't even matter what is the origin of randomness in step (2), provided that the sampled circuits (Hamiltonians) are independent and correctly distributed, and known to the experimenter by the end of the experiment. 
 
For both the Ising model ($\mathcal{C}_I$) and degree-3 polynomial ($\mathcal{C}_f$) case the most obvious experimental implementations require non-local gates, which is quite challenging. In terms of the Ising model case 2-qubit interactions might need to be generated between any two qubits in a system. Such interactions do not arise naturally in lattice geometries and as such will need to be engineered. From an experimental perspective there has already been significant progress in these directions. The dynamics of the Ising model with local interactions have been digitally simulated in ion traps \cite{lanyon11,britton12} and very recently non-local interactions have been utilised in the digital simulation of fermionic systems with superconductors \cite{barends15}. As technologies such as cavity buses for superconducting systems \cite{majer07} become more reliable, we expect that an increasing number of systems will be able to implement IQP circuits in a regime that is likely not to be classically simulable. 

Theoretically there are a number of natural questions that remain to be answered, the most obvious of which is whether or not Conjectures \ref{con:Ising} and \ref{con:deg3} are true. To the best of our knowledge there are no known techniques that are capable of lifting the known worst-case hardness results for the Ising model and gaps of degree-3 polynomials to average-case results. However, recent breakthroughs \cite{sly12,sinclair14,galanis14} in categorising the complexity of statistical mechanical systems via the underlying interaction graph properties give some hope that these conjectures can be resolved. 

As well as the connections used here between IQP, the Ising model and low-degree polynomials, it is known that IQP circuits are closely related to Tutte polynomials and weight enumerator polynomials of binary linear codes~\cite{shepherd10}. It is a compelling question as to whether the ideas developed here can be extended to these models. It would also be very interesting to explore the connection between nonadaptive measurement-based quantum computing and IQP~\cite{browne11,hoban14} in order to find subclasses of IQP that are physically straightforward to implement, and for which one can prove conditional hardness results similar to those shown here.

% ------------------------------------------------------------------------------

\begin{acknowledgements}
We would like to thank Aram Harrow, Richard Jozsa, Gavin Brennen, Steve Flammia, and Peter Rohde for helpful discussions and comments on this manuscript. AM was supported by the UK EPSRC under Early Career Fellowship EP/L021005/1. MJB has received financial support from the Australian Research Council via the Future Fellowship scheme (grant FT110101044), and Lockheed Martin Corporation. % We would also like to thank Y. Henrys for motivation in the final stages of this work.
\end{acknowledgements}

% ------------------------------------------------------------------------------

\bibliography{rsr2}

\onecolumngrid
\appendix
%\large

%------------------------------------------------------------------------

\section{The Ising model and IQP circuits}
\label{appendix:Ising}
Consider the general Ising model:
\[ H_I= \sum_{i<j} w_{ij} z_i z_j +\sum_{i} v_i z_i \]
where  $z \in \{-1,+1\}^n$,  $i,j$ label vertices in the complete graph, $w_{ij}$ is a weight assigned to the edge $(i,j)$ and $v_i$ is a weight assigned to the vertex $i$. We want to show that the partition function of $H_I$  evaluated at $e^{i \theta}$, $Z= \text{Tr}[e^{i\theta H_I}]$, is proportional to an amplitude of an $n$-qubit IQP circuit. 

Imagine we have an IQP circuit constructed from the gates $e^{i\theta X\otimes X}$ and $e^{i\theta X}$, and consider the circuit amplitude
\[\langle 0|^{\otimes n} e^{i\theta(\sum_{i<j} w_{ij} X_i X_j +\sum_{i} v_i X_i)} |0\rangle^{\otimes n}\]
where X is a Pauli-X operator. Then
\bea
\langle 0|^{\otimes n} e^{i\theta(\sum_{i<j} w_{ij} X_i X_j +\sum_{i} v_i X_i)} |0\rangle^{\otimes n} & = &\langle 0|^{\otimes n} H^{\otimes n}e^{i \theta(\sum_{i<j} w_{ij} Z_i Z_j +\sum_{i} v_i Z_i} H^{\otimes n} |0\rangle^{\otimes n} \nonumber \\
& = &\frac{1}{2^n}\sum_{x,y \in\{0,1\}^n} \langle y| e^{i \theta(\sum_{i<j} w_{ij} Z_i Z_j +\sum_{i} v_i Z_i)} |x\rangle \nonumber  \\
& = &\frac{1}{2^n}\sum_{x\in \{0,1\}^n} e^{i \theta(\sum_{i<j} w_{ij} (-1)^{x_i x_j} +\sum_{i} v_i (-1)^{x_i})} \nonumber \\
& =  &\frac{1}{2^n}\sum_{z\in \{-1,1\}^n} e^{i\theta(\sum_{i<j} w_{ij} z_i z_j +\sum_{i} v_i z_i)} \nonumber \\
& = & \frac{1}{2^n} \text{Tr}[e^{i\theta H_I}]. \nonumber
\eea
If we consider IQP circuits constructed from the gates $e^{i\frac{\pi}{t}X\otimes X}$ and $e^{i\frac{\pi}{t}X}$, then it is clear that amplitudes of such circuits give rise to partition functions of the form $Z(\omega)=\text{Tr}[\omega^{H_I}]$ where $\omega=e^{i\frac{\pi}{t}}$ and $w_{ij},v_i\in \{0,\dots,2t-1\}$.

See the final paragraph of Appendix \ref{sec:tech} below for an alternative connection between IQP and the Ising model using a different gate set.

%-------------------------------------------------------------

\section{Gaps of degree-3 polynomials}
\label{sec:deg3}
We can also use IQP circuit amplitudes to express the gap of any degree-3 polynomial over $\mathbb{F}_2$ (with no constant term), $f:\{0,1\}^n\rightarrow \{0,1\}$:
\[ f(x) = \sum_{i,j,k} \alpha_{ijk} x_i x_j x_k + \sum_{i,j} \beta_{ij} x_i x_j + \sum_i \gamma_i x_i, \]
where $\alpha_{ijk}, \beta_{ij}, \gamma_i \in \{0,1\}$. 

This can be seen by noting the action of the Z, CZ, and CCZ gates on the computational basis states. Given a single qubit in a computational basis state, $|x_1\rangle$, the Pauli-Z gate has action $Z|x_1\rangle = (-1)^{x_1}|x_1\rangle$. Similarly we see $CZ_{12}|x_1x_2\rangle = (-1)^{x_1x_2}|x_1x_2\rangle$ and $CCZ_{123}|x_1x_2x_3\rangle = (-1)^{x_1x_2x_3}|x_1x_2x_3\rangle$. Then, in order to generate the phase $(-1)^{f(x)}|x\rangle$, we simply apply for each non-zero $\alpha_{ijk}, \beta_{ij}, \gamma_i$ in $f$ a corresponding CCZ, CZ, or Z gate to the $n$-qubit computational basis state, $|x\rangle$. Let $\mathcal{\widetilde{C}}_f$ denote this circuit and $\mathcal{C}_f = H^{\otimes n}\mathcal{\widetilde{C}}_fH^{\otimes n}$ denote its Hadamard transform, which is an IQP circuit.

We want to show that $\bra{0}^{\otimes n} \mathcal{C}_f \ket{0}^{\otimes n}$ is proportional to the gap of $f$. This is straightforward:
\[ \langle 0|^{\otimes n} H^{\otimes n} \mathcal{\widetilde{C}}_f H^{\otimes n}|0\rangle^{\otimes n} = \frac{1}{2^n} \sum_{x,y\in\{0,1\}^n}\langle y| \mathcal{\widetilde{C}}_f|x\rangle = \frac{1}{2^n} \sum_{x\in\{0,1\}^n} (-1)^{f(x)} = \frac{1}{2^n}\gap(f), \]
recalling that $\gap(f)$ is defined as $\gap(f) := |\{x: f(x)=0\}| - |\{x: f(x)=1\}|$. Note that similar ideas were previously used in~\cite{dawson05,rudolph09} for different classes of circuits.

% ------------------------------------------------------------------------------------

\section{The Hadamard gadget}
\label{sec:hadgadget}

Arbitrary Hadamard operations can be implemented in IQP using post-selection~\cite{bremner11} -- something that cannot be done efficiently in physical systems but nonetheless turns out to be useful as a mathematical tool. We perform Hadamard gates by introducing \emph{Hadamard gadgets} to our IQP circuits. Imagine that we have an $n$-qubit state $|\psi\rangle$ upon which we wish to apply a Hadamard gate to qubit $j$, and denote this by $H_j|\psi\rangle$. We now add an extra qubit $e$ initialised to the state $H_e|0\rangle_e$, recalling that all lines of an IQP circuit in the Z basis begin and end with a H. Then we perform some gates to produce the state $H_jCZ_{je}H_e|\psi\rangle|0\rangle_e$ and measure qubit $j$ in the computational basis, post-selecting on the $0$ outcome. We note that the gate $H_j$ must be present as this is an IQP circuit. It is easily checked that the effect of this protocol is to teleport the state of qubit $j$ to qubit $e$ while performing a Hadamard operation on this qubit. If we had desired to perform more gates on qubit $j$ we now perform them on qubit $e$. 

There are a few important things to note about the Hadamard gadget. Firstly for each ``intermediate" Hadamard to be performed there is a cost of 1 qubit, 1 CZ gate, and 1 qubit of post-selection. Secondly, if we did not perform post-selection there would be a probability of $1/2$ for this procedure to implement a Hadamard gate, so if a circuit had $m$ intermediate Hadamard gates the success probability would be $1/2^m$, making this procedure impractical. 

Finally, it is clear that the Hadamard gadget \emph{preserves multiplicative approximations of circuit amplitudes}. Imagine we had an $n$-qubit quantum circuit $\mathcal{U}$ which is expressed in terms of gates that are diagonal in the Z basis and $m$ intermediate $H$ gates. Then using the Hadamard gadget $m$ times we get the IQP circuit $\mathcal{C}$ such that $\langle0|^{\otimes n}\mathcal{U}|0\rangle^{\otimes n}=2^{m/2}\langle 0|^{\otimes(n+m)}\mathcal{C}|0\rangle^{\otimes(n+m)}$. Furthermore, if we had a multiplicative approximation, $A$, to this latter quantity within multiplicative error $\gamma$ then $|A-\langle 0|^{\otimes(n+m)}\mathcal{C}|0\rangle^{\otimes(n+m)}|\leq \gamma \langle 0|^{\otimes(n+m)}\mathcal{C}|0\rangle^{\otimes(n+m)}$. Then it is clear that $A' = 2^{m/2}A$ also provides a multiplicative approximation to $\langle0|^{\otimes n}\mathcal{U}|0\rangle^{\otimes n}$ within $\gamma$.

%-------------------------------------------------------------------------------

\section{Hardness of multiplicative approximations}
\label{sec:multiplicative}
We now show that approximating the squared gap of a degree-3 polynomial over $\F_2$ is $\GapP$-complete, even up to constant multiplicative error. $\GapP$ is a complexity class containing \#P~\cite{fenner94}. The majority of the proof is very similar to arguments used in other settings, such as for the permanent~\cite{aaronson11a} or the Ising model~\cite{goldberg14}. However, we believe that the connection to IQP again somewhat simplifies the argument. A similar argument could be used to prove hardness of multiplicative approximation for the Ising model; we omit this as it was shown in~\cite{goldberg14}.

\begin{prop}
Let $f$ be a degree-3 polynomial over $\F_2$. The problem of producing an estimate $\widetilde{z}$ such that $||\ngap(f)| - \widetilde{z}| \le \epsilon|\ngap(f)|$, for any $\epsilon < 1/2$, is $\GapP$-complete.
\end{prop}

\begin{proof}
Let $f:\{0,1\}^n \rightarrow \{0,1\}$ be a function described by a polynomial-size classical circuit $C$. By definition, computing $\ngap(f)$ for arbitrary functions $f$ is $\GapP$-complete~\cite{fenner94}. We now show that this problem can be reduced to approximating $|\ngap(f)|$ for degree-3 polynomials $f$, up to sufficiently high precision.

The first step is to show a connection between a ``shifted'' version of the general problem and IQP. For any $m = O(n)$ and for any $c \in [-1,1]$ such that $c$ is an integer multiple of $2^{1-m}$, there is a polynomial-size classical circuit $D_c$ on $m$ input bits computing a function $g$ such that $\ngap(g) = -c$. By adding a control bit which determines whether to execute $C$ or $D_c$, we obtain a polynomial-size classical circuit on $n+m+1$ bits computing a function $f_c$ such that $\ngap(f_c) = \frac{1}{2}(\ngap(f) - c)$.

By universality of the Toffoli gate for reversible classical computation, we can write down a quantum circuit $Q_c$ consisting of Toffoli, CNOT and X gates such that $Q_c$ performs the map
\[ \ket{0}^{\otimes a} \ket{x} \ket{0} \mapsto \ket{0}^{\otimes a} \ket{x} \ket{f_c(x)}, \]
where the first register contains $a$ ancilla qubits, with $a = \poly(n)$. Let the circuit $Q'_c$ be defined as follows: start in the state $\ket{0}^{\otimes a} \ket{0}^{\otimes (n+m+1)} \ket{0}$; apply an X gate to the third register; apply Hadamard gates to the second and third registers; execute $Q_c$; apply Hadamard and X gates to the third register; and finally apply Hadamard gates to the second register. The state near the end of this circuit, just before the last Hadamard gates are applied, is
\[ \ket{0}^{\otimes a} \left( \frac{1}{\sqrt{2^{n+m+1}}} \sum_{x \in \{0,1\}^{n+m+1}} (-1)^{f_c(x)} \ket{x} \right) \ket{0}. \]
Therefore, $\bracket{0}{Q'_c}{0} = \ngap(f_c) = \frac{1}{2}(\ngap(f) - c)$, where we henceforth abbreviate $\ket{0}^{\otimes n} = \ket{0}$ for any integer $n$. As discussed in Appendix \ref{sec:hadgadget}, there exists an IQP circuit $Q''_c$ such that $\bracket{0}{Q'_c}{0} = \alpha \bracket{0}{Q'_c}{0}$ for some easily computed real number $\alpha$. An approximation of $\bracket{0}{Q''_c}{0}$ up to multiplicative error $\epsilon$ therefore implies an approximation of $\bracket{0}{Q'_c}{0}$ up to the same error.

Assume we have a deterministic algorithm $\mathcal{A}$ such that, on taking as input an arbitrary IQP circuit $C$, $\mathcal{A}$ outputs $\widetilde{z}$ such that $||\bracket{0}{C}{0}| - \widetilde{z}| \le \epsilon |\bracket{0}{C}{0}|$, for some $\epsilon<1$. We will use $\mathcal{A}$ to output estimates $\widetilde{d}$ such that $|\widetilde{d} - |\ngap(f) - c|| \le \epsilon |\ngap(f) - c|$. Observe that $\mathcal{A}$ can be used to distinguish with certainty between the two cases that $c = \ngap(f)$, and $c \neq \ngap(f)$. So if we can guess $c = \ngap(f)$, $\mathcal{A}$ certifies that our guess is correct.

We now perform the following sequence of steps to compute $\ngap(f)$. At each step, we have a guess $c$ for $\ngap(f)$, starting with $c = 0$. We use $\mathcal{A}$ to output an estimate $\widetilde{d}$ of $|\ngap(f) - c|$. We then consider the two possibilities $c\pm\widetilde{d}$. We apply $\mathcal{A}$ to estimate $|\ngap(f) - (c \pm \widetilde{d})|$, obtaining outputs $\widetilde{d}^+$, $\widetilde{d}^-$, and choose as our new value of $c$ either $c + \widetilde{d}$ (if $\widetilde{d}^+ \le \widetilde{d}^-$) or otherwise $c - \widetilde{d}$. We then repeat until $c = \ngap(f)$.

Write $\widetilde{d} = (1+\gamma) |c-\ngap(f)|$, for some $\gamma$ such that $|\gamma| \le \epsilon$. If $c < \ngap(f)$, we have
\[ |c + \widetilde{d} - \ngap(f)| = |c + (1+\gamma)(\ngap(f)-c) - \ngap(f)| = |\gamma(\ngap(f)-c)| \le \epsilon|c-\ngap(f)|, \]
so if we choose $c + \widetilde{d}$ for the next step, the distance between $c$ and $\ngap(f)$ is reduced by multiplying it by a factor smaller than $\epsilon$. A similar calculation holds in the case $c > \ngap(f)$ if we choose $c-\widetilde{d}$ for the next step. As $\ngap(f)$ is an integer multiple of $2^{-n}$, if we pick the correct choice at each step, we need only $O(n)$ steps to obtain $c = \ngap(f)$. We now show that, for sufficiently small $\epsilon$, the correct choice is indeed picked.

Assuming $c < \ngap(f)$, we would like $\widetilde{d}^+ < \widetilde{d}^-$. This will hold if
\beas
&& (1+\epsilon)|c + \widetilde{d} - \ngap(f)| < (1-\epsilon) |c - \widetilde{d} - \ngap(f)|\\
&\Leftrightarrow& (1+\epsilon)| c + (1+\gamma) |c-\ngap(f)| - \ngap(f) | < (1-\epsilon)| c - (1+\gamma) |c-\ngap(f)| - \ngap(f) |\\
%&\Leftrightarrow& (1+\epsilon)| c + (1+\gamma)(\ngap(f)-c) - \ngap(f) | < (1-\epsilon)| c - (1+\gamma)(\ngap(f)-c) - \ngap(f) |\\
&\Leftrightarrow& (1+\epsilon)|\gamma(\ngap(f)-c)| < (1-\epsilon)|(2+\gamma)(\ngap(f)-c)|\\
&\Leftrightarrow& (1+\epsilon)|\gamma| < (1-\epsilon)|2+\gamma|.
\eeas
This inequality is easily seen to hold for all $\gamma$ such that $|\gamma| \le \epsilon < 1/2$. The case $c > \ngap(f)$ is similar.

One omission in this argument is that we can only approximate $|\ngap(f) - c|$ for values $c$ which are integer multiples of $2^{1-m}$, but depending on the estimates output by $\mathcal{A}$, we may want to use values of $c$ which are not integer multiples. In this case, we simply truncate $c$ to the nearest integer multiple of $2^{1-m}$. For sufficiently large $m = O(n)$, the additional error introduced by this truncation is negligibly small.
\end{proof}

This immediately implies hardness of multiplicatively approximating $\ngap(f)^2$ to the same accuracy, because
\[ |\ngap(f)^2 - \widetilde{z}^2| \le \epsilon\ngap(f)^2 \Rightarrow ||\ngap(f)| - \widetilde{z}|(|\ngap(f)| + \widetilde{z}) \le \epsilon\ngap(f)^2 \Rightarrow ||\ngap(f)| - \widetilde{z}| \le \epsilon |\ngap(f)|. \]
 
 %------------------------------------------------------------------------
 
\section{From additive to multiplicative error}
\label{sec:addmult}

An essential ingredient of our argument connecting sampling with additive and multiplicative error (Lemma \ref{lem:approx}) is a classic result of complexity theory, the so-called Stockmeyer Counting Theorem. Informally it states that for every function inside \#P, there is a good multiplicative approximation within $\FBPP^{\NP}$:

\begin{thm}[Stockmeyer~\cite{stockmeyer85}, see~\cite{aaronson13} for statement here]
\label{thm:stockmeyer}
There exists an $\FBPP^{\NP}$ machine which, for any boolean function $f:\{0,1\}^n \rightarrow \{0,1\}$, can approximate
\[ p = \Pr_x [ f(x) = 1 ] = \frac{1}{2^n} \sum_{x \in \{0,1\}^n} f(x) \]
to within a multiplicative factor of $(1+\epsilon)$ for all $\epsilon = \Omega(1/\poly(n))$, given oracle access to $f$.
\end{thm}

We also include here a corollary of Lemma \ref{lem:approx}, which formalises the discussion in the main text.

\begin{cor}
\label{cor:addapprox}
Let $\mathcal{F}$ be a family of IQP circuits on $n$ qubits. Pick a random circuit $\mathcal{C}$ by choosing a circuit from $\mathcal{F}$ uniformly at random, then appending X gates on a uniformly random subset of the qubits. Assume that there exist universal constants $\alpha,p>0$ such that $\Pr[|\bracket{0}{\mathcal{C}}{0}|^2 \ge \alpha\cdot2^{-n}] \ge p$. Further assume there exists a classical polynomial-time algorithm $\mathcal{A}$ which, for any IQP circuit $\mathcal{C}'$, can sample from a probability distribution which approximates the output probability distribution of $\mathcal{C}'$ up to additive error $\epsilon = \alpha p / 8$ in $\ell_1$ norm. Then there is a $\FBPP^{\NP}$ algorithm which, given access to $\mathcal{A}$, approximates $|\bracket{0}{\mathcal{C}}{0}|^2$ up to multiplicative error $1/4 + o(1)$ on at least a $p/2$ fraction of circuits $\mathcal{C}$.
\end{cor}

\begin{proof}
With probability at least $p$ over the choice of $\mathcal{C}$, $|\bracket{0}{\mathcal{C}}{0}|^2 \ge \alpha \cdot 2^{-n}$.
Setting $\epsilon = \alpha p / 8$, $\delta = p/2$ in Lemma \ref{lem:approx}, there is a $\FBPP^{\NP}$ algorithm which, given access to $\mathcal{A}$, approximates $|\bracket{0}{\mathcal{C}}{0}|^2$ up to additive error $O((1+o(1))\alpha \cdot 2^{-n-2} + |\bracket{0}{\mathcal{C}}{0}|^2/\poly(n))$ with probability at least $1-p/2$ (over the choice of $\mathcal{C}$). Therefore, with probability at least $p/2$, the algorithm approximates $|\bracket{0}{\mathcal{C}}{0}|^2$ up to multiplicative error $1/4 + o(1)$.
\end{proof}

%------------------------------------------------------------------------
 
\section{Proofs of anticoncentration bounds}
\label{sec:tech}

Here we prove the lemmas required for the anticoncentration bounds claimed in the main body of the paper. Our main technical tool will be the following result:

\begin{lem}
\label{lem:tech}
Let $\omega = e^{2 \pi i / r}$, $\eta = e^{2 \pi i / s}$ for some integers $r$ and $s$ such that $r,s\ge 2$, and if $s=2$ then $r=2$. Fix integer $n$. For pairs of integers $i<j$ between 1 and $n$, let $\alpha_{ij}$ be picked uniformly at random from $\{0,\dots,r-1\}$. For $k \in \{1,\dots,n\}$, let $\beta_k$ be picked uniformly at random from $\{0,\dots,s-1\}$. Then
\[ \sum_{w,x,y,z \in \{0,1\}^n} \left| \E_{\alpha} \left[ \omega^{\sum_{i<j} \alpha_{ij}(w_i w_j + x_i x_j - y_i y_j - z_i z_j)} \right] \E_\beta \left[\eta^{\sum_k \beta_k(w_k+x_k - y_k - z_k)} \right]\right| \le 3 \cdot 2^{2n}. \]
\end{lem}

\begin{proof}
Each term in the sum is clearly upper-bounded by 1. Our goal will be to show that many terms in the sum are actually zero, using the identities $\E_\alpha[\omega^{\alpha t}] = 0$ for any integer $t \not\equiv 0$ mod $r$, and similarly $\E_\beta[\eta^{\beta t}] = 0$ for $t \not\equiv 0$ mod $s$. First,
\[ \E_\beta\left[\eta^{\sum_k \beta_k(w_k+x_k - y_k - z_k)} \right] = \prod_{k=1}^n \E_{\beta_k}[\eta^{\beta_k(w_k+x_k - y_k - z_k})], \]
which equals 0 unless $z \equiv w + x - y$ mod $s$. If this is the case, it also holds that $z \equiv w + x - y$ mod $r$. To see this, observe that for $s \ge 3$, as $w+x-y \in \{-1,0,1,2\}^n$, $z \equiv w + x - y$ mod $s$ implies that in fact $z = w + x - y$. If $s = 2$, then $r = s$ by assumption and the claim is immediate. Hence
\beas
\E_{\alpha} \left[ \omega^{\sum_{i<j} \alpha_{ij}(w_i w_j + x_i x_j - y_i y_j - z_i z_j)} \right]  &=& \E_{\alpha} \left[ \omega^{\sum_{i<j} \alpha_{ij}(w_i w_j + x_i x_j - y_i y_j - (w+x-y)_i (w+x-y)_j)} \right] \\
&=& \prod_{i<j} \E_{\alpha_{ij}} \left[ \omega^{\alpha_{ij}(w_i (y_j-x_j)+x_i(y_j-w_j)+y_i(w_j+x_j)-2y_iy_j)} \right].
\eeas
We claim that if $w$, $x$ and $y$ are all distinct, then there exists a pair $(i,j)$ (also distinct) such that
\be \label{eq:mod}  w_i(y_j-x_j) + x_i(y_j-w_j) + y_i(w_j+x_j)  \not\equiv 2y_iy_j  \text{ mod } r, \ee
which implies that the product evaluates to 0. Since $w \not= x$, let $i$ be such that $w_i \not= x_i$.  We may assume without loss of generality that $w_i=0$ and $x_i=1$, since the other case just relabels $w$ for $x$ and vice versa.  Therefore it remains to find a $j>i$ such that
\be \label{eq:mod2}  y_j - w_j + y_i(w_j+x_j)  \not\equiv 2y_iy_j  \text{ mod } r.  \ee
In fact, finding a $j<i$ is just as good, because the expression (\ref{eq:mod}) is also symmetric under exchange of $i$ and $j$. If $y_i=0$ then since $y \not= w$, there must exist some $j$ such that $y_j \not= w_j$.  Moreover, $j \not= i$ because $y_j$ cannot differ from both $w_i$ and $x_i$ (which we have already seen are different).  So $y_j = 0$, $w_j = 1$ (or vice versa), they cannot be equivalent modulo $r$, and (\ref{eq:mod2}) holds in this case. On the other hand, if $y_i=1$ then since $y \not= x$, there must be some $j$ such that $y_j \not= x_j$, and $j \not= i$ because we already chose $x_i = 1$.  So $y_j = 0$, $x_j = 1$ (or vice versa), they are not equivalent modulo $r$, and (\ref{eq:mod2}) holds in this case too.

%We claim that, if $w \neq x \neq y$, there exists a pair of distinct indices $(i,j)$ such that
%
%\[ g_{ij} := w_i (y_j-x_j)+x_i(y_j-w_j)+y_i(w_j+x_j) \not\equiv 2y_iy_j \text{ mod } r, \]
%
%which implies that the whole product evaluates to 0. As $y_i y_j \in \{0,1\}$, to prove this claim it suffices to show that $g_{ij} = \pm1$. Let $i$ be such that $x_i \neq y_i$ and assume without loss of generality that $x_i = 0$, $y_i = 1$. Then, for each $j$, $g_{ij} = w_i (y_j-x_j) + w_j + x_j$. If $w_i = 0$, let $j$ be such that $w_j \neq x_j$; if $w_i = 1$, let $j$ be such that $w_j \neq y_j$. Such a position $j$ must exist in each case because $w \neq x$ and $w \neq y$, and in either case $j \neq i$. In the first case, $g_{ij} = w_j + x_j = 1$; in the second case, $g_{ij} = y_j  + w_j = 1$.

We have therefore shown that each term in the sum in the statement of the lemma evaluates to 0 unless both $z \equiv w + x - y$ mod $s$, and also $w=x$, $w=y$ or $x=y$. The number of ways of choosing strings $w$, $x$, $y$, $z$ to comply with this constraint is at most $3 \cdot 2^{2n}$. This completes the proof.
\end{proof}

% -------------------------------------------

We first apply Lemma \ref{lem:tech} to random degree-3 polynomials. A degree-3 polynomial $f:\{0,1\}^n \rightarrow \{0,1\}$ can be picked uniformly at random in a straightforward manner. Assuming that $f(0^n) = 0$ without loss of generality and writing
\[ f(x) = \sum_{i,j,k} \alpha_{ijk} x_i x_j x_k + \sum_{i,j} \beta_{ij} x_i x_j + \sum_i \gamma_i x_i, \]
where $\alpha_{ijk}, \beta_{ij}, \gamma_i \in \{0,1\}$, we simply pick each of these coefficients uniformly at random. Observe that the random $\gamma_i$ coefficients correspond to applying X gates to a random subset of the qubits. We now show that, for each possible choice of the $\alpha_{ijk}$ coefficients, the expected value of $\ngap(f)^4$ (taken over the $\beta$, $\gamma$ coefficients) is low.

\begin{lem}
\label{lem:degle2random}
Let $f:\{0,1\}^n \rightarrow \{0,1\}$ be a polynomial over $\F_2$ whose degree $\le 2$ part is uniformly random. Then $\E_f[\ngap(f)^4] \le 3 \cdot 2^{-2n}$.
\end{lem}

\begin{proof}
We have
\beas
\E_f\left[\ngap(f)^4\right] &=& \E_f\left[ \left(\E_{x \in \{0,1\}^n} (-1)^{f(x)}\right)^4 \right]\\
&=& \E_f\left[ \E_{w,x,y,z \in \{0,1\}^n} (-1)^{f(w)+f(x)+f(y)+f(z)} \right]\\
  &=& \E_{w,x,y,z \in \{0,1\}^n} \E_f  \left[(-1)^{f(w)+f(x)+f(y)+f(z)} \right].
\eeas
Let $f_{\le 2}$ and $f_{>2}$ be the parts of $f$ of degree $\le 2$ and degree $>2$, respectively. Then
\[ \E_f  \left[(-1)^{f(w)+f(x)+f(y)+f(z)} \right] = (-1)^{f_{>2}(w) + f_{>2}(x) + f_{>2}(y) + f_{>2}(z)} \E_{f_{\le 2}} \left[(-1)^{f_{\le 2}(w)+f_{\le 2}(x)+f_{\le 2}(y)+f_{\le 2}(z)} \right]. \]
Expanding $f_{\le 2}(x) = \sum_{i,j} \beta_{ij} x_i x_j + \sum_k \gamma_k x_k$, the expectation in this expression can be written as
\[ \E_{\beta,\gamma} \left[ (-1)^{\sum_{i,j} \beta_{ij}(w_i w_j + x_i x_j + y_i y_j + z_i z_j) + \sum_k \gamma_k(w_k+x_k + y_k + z_k)} \right]. \]
So
\beas \E_f\left[\ngap(f)^4\right] &\le& \frac{1}{2^{4n}}\sum_{w,x,y,z \in \{0,1\}^n}  \left|\E_{\beta}\left[ (-1)^{\sum_{i,j} \beta_{ij}(w_i w_j + x_i x_j + y_i y_j + z_i z_j)}\right] \E_\gamma\left[(-1)^{\sum_k \gamma_k(w_k+x_k + y_k + z_k)} \right] \right|\\
&\le& 3\cdot 2^{-2n}
\eeas
by Lemma \ref{lem:tech}.
\end{proof}

Next we consider the partition function (\ref{eq:partition}) of the Ising model for random weights $w_{ij},v_k \in \{0,\dots,7\}$, evaluated at $\omega =  e^{i\pi / 8}$. To make contact with the previous analysis, and with IQP circuits, we rewrite the expression (\ref{eq:partition}) for $Z(\omega)$ as
\beas
Z(\omega) &=& \sum_{x \in \{0,1\}^n} \omega^{\sum_{i<j} w_{ij} (1-2x_i)(1- 2x_j) + \sum_{k=1}^n v_k (1-2x_k) }\\
&=& \omega^{\sum_{i<j} w_{ij} + \sum_{k=1}^n v_k} \sum_{x \in \{0,1\}^n} \omega^{4\sum_{i<j} w_{ij} x_i x_j + 2 \sum_{k=1}^n v'_k x_k }, \\
\eeas
%
%where each $w_{ij}$ is picked uniformly at random from $\{0,\dots,15\}$ and
where $v'_k = -v_k - \sum_{j\neq k} w_{jk}$. Observe that $v'_k$ is uniformly random mod 8. Up to an easily computed phase, we can write
\[ Z(\omega) = \sum_{x \in \{0,1\}^n} i^{\sum_{i<j} w_{ij} x_i x_j} e^{(\pi i / 4)( \sum_{k=1}^n v'_k x_k) }. \]
Then it is clear that $Z(\omega)/2^n = \bra{0}^{\otimes n}\mathcal{C}_I\ket{0}^{\otimes n}$ for an IQP circuit $\mathcal{C}_I$ consisting of gates from the set $\{\diag(1,1,1,i), \diag(1,e^{i\pi / 4})\}$. Further, because of the uniformly random choice of $w$ and $v'$, we can consider such a circuit as being formed by picking a random circuit using this gate set, then following it by a random choice of X gates. It follows immediately from Lemma \ref{lem:tech} that $\E_{w,v'} \left[|Z(\omega)|^4 \right] \le 3\cdot 2^{2n}$. Finally, note that we could in fact have picked the $w_{ij}$ weights uniformly from the set $\{0,\dots,3\}$ without affecting the above analysis.

% ------------------------------------------------------------------------------

\end{document}